\documentclass[a4paper,UKenglish]{lipics-v2016}

\usepackage{amsthm,amsmath,graphicx}
\usepackage{bm}

\DeclareMathOperator{\val}{pty}
\DeclareMathOperator{\bin}{bin}

\DeclareMathOperator{\dom}{dom}
\DeclareMathOperator{\im}{im}
\DeclareMathOperator{\dome}{dom_{\text{even}}}

\newcommand{\ZZ}{\mathbb{Z}}
 \newtheorem{claim}{Claim}

\begin{document}
	
\bibliographystyle{plainurl}

\title{A short proof of correctness of the quasi-polynomial time algorithm for parity games}

\author[1]{Hugo Gimbert}
\affil[1]{LaBRI, Universit{\' e} de Bordeaux, CNRS, France
  \texttt{hugo.gimbert@cnrs.fr}}
\author[2]{Rasmus Ibsen-Jensen}
\affil[2]{IST Austria, Vienna, Austria \texttt{ribsen@ist.ac.at}}

\authorrunning{H. Gimbert, R. Ibsen-Jensen}




\maketitle

\begin{abstract}
Recently Cristian S. Calude, Sanjay Jain, Bakhadyr Khoussainov, Wei Li and Frank Stephan proposed a quasi-polynomial time algorithm for parity games~\cite{calude}. These notes provide a short proof of correctness of their algorithm.
\end{abstract}

\paragraph*{Parity games}
A parity game is given by a directed graph $(V, E)$, a starting node $s \in V$,
 a function  which attaches to each $v \in V$ a priority $\val(v)$ from a set $\{1,2,...,m\}$;
 the main parameter of the game is $n$, the number of nodes, and the second parameter is $m$. Two players Anke and Boris move alternately in the graph with Anke moving first. A move from a node $v$ to another node $w$ is valid if $(v, w)$ is an edge in the graph; furthermore, it is required that from every node one can make at least one valid move. The alternate moves by Anke and Boris define an infinite sequence of nodes which is called a play.  Anke wins a play through nodes $v_0, v_1,\cdots$ iff $\limsup_t \val(v_t)$ is even,
otherwise Boris wins the play. 

  We say that a player \emph{wins the parity game} if she has a strategy which guarantees the play to be winning for her. Parity games are determined~\cite{zielonka:1998} thus either Anke or Boris wins the parity game.
\paragraph*{Statistics}

The core of the algorithm of Calude et al. is to keep track of statistics about the game,
in the form of \emph{partial} functions
\[
f : 0\ldots k \to 1\ldots m\enspace.
\]
The integer $k$ is chosen such that $2^k$
is strictly larger than twice the number of vertices. The domain of $f$ is denoted $\dom(f)$
and its image $\im(f)$. We also let $\dome(f)=\{f(i)\text{ is even}|i\in \im(f)\}$.
Statistics are assumed to be \emph{increasing},
i.e.
 $
\forall i,j \in\dom(f), (i \leq j \implies f(i) \leq f(j))\enspace$.
A statistic $f$ can be modified by \emph{inserting} a priority $c$ at an index $\ell$,
which results in removing all pairs of index $\leq \ell$ from $f$ and adding the pair $(\ell,c)$.

The initial statistic is the empty statistic $f_0=\emptyset$,
which is updated successively by all the priorities visited during the play, thus producing a sequence of statistics.
The update of a statistic $f$ by a priority $c$ is performed by applying successively the following two rules.
\begin{itemize}
\item {\bf Type I update:} If $c$ is even then it is inserted
at the highest index $j\in 0\cdots k$ such that 
$f$ is defined and even on $0 \ldots j-1$.
\item {\bf Type II update:} If $\im(f)$ contains at least one value $<c$ then $c$ is inserted at the highest index $j\in \dom(f)$ such that $f(j) <c$.
\end{itemize}

Applying both rules in succession ensures that the update of an increasing statistic is increasing.
If rule II triggers an insertion then we say the update is a type II update.
Notice that in this case, applying or not rule I in the first place does not change the result.
If rule I triggers an insertion but rule II does not then we say the update is a type I update.

Anke (resp. Boris) \emph{wins the statistics game} if she (resp. he) has a strategy to enforce (resp. to avoid) a visit to a statistic whose domain contains $k$.
Similarly to the game of chess, statistics games are determined: either Anke or Boris has a winning strategy~\cite{zermelo}.

\section{Correctness of the algorithm}

\begin{theorem}[Calude et al]\label{theo:main}
Anke wins the parity game iff she wins the statistics games.
\end{theorem}

Since statistics games are determined, the direct implication follows from:

\begin{lemma}\label{lem:boris}
If Boris wins the statistics games, he wins the parity game.
\end{lemma}

\begin{proof}
Every play won by Boris in the statistics game is won by Boris in the parity game
because
$c=\limsup_t c_t$ is odd
in
every sequence
of statistics updates
$f_0\to_{c_0} f_1 \to_{c_1} \ldots $
such that $\forall t \geq 0, k\not \in\dom(f_t)$,
the proof of which follows.

An easy case is when the sequence of statistics
is ultimately constant equal to some $f$
then $f\to_c f$
thus $c$ is odd because an update by an even priority
always performs an insertion.
In the opposite case define
$(\ell,d)$
the maximal pair (for the dichotomic order)
inserted infinitely often.
Since $d$ is inserted infinitely often then $d \leq \limsup_t c_t = c$.
And $d\geq c$  
otherwise $c$ would be inserted infinitely often at an  index $\geq \ell$ which would contradict the maximality of $(\ell,d)$.
Since $(\ell,c)$ is inserted infinitely often 
then it is removed infinitely often,
let $(\ell',c')$ be a pair used infinitely often to remove $(\ell,c)$.
Since $(\ell,c)$ is removed by the insertion of  $(\ell',c')$ then $\ell \leq \ell'$
hence $\ell=\ell'$ by maximality of $\ell$.
Since $(\ell',c')=(\ell,c')\neq (\ell,c)$
then $c' < c$ by maximality of $c$.
Since $(\ell,c)$ is removed by the insertion of  
$(\ell',c')$ and $c' < c$ then this insertion is necessarily performed by a type I update
and since $\ell=\ell'$ then $c$ is odd
otherwise the insertion would occur at a higher index.
\end{proof}

The converse implication (Corollary~\ref{cor:direct})
relies on several crucial  properties of  statistics.

\begin{definition}[Counter value]
With every statistic $f$ is associated its \emph{counter value}
\[
\bin(f) =
\sum_{j \in\dome(f)} 2^j
\enspace.
\]
\end{definition}

In the sequel we fix a sequence  $f_0\to_{c_0} f_1  \ldots \to_{c_N} f_{N+1}$ of statistics updates.
We will first give two lemmas that gives information on what can be said when an update of type 1 and 2, respectively, is used on a date.

\begin{lemma}\label{lem:rule1}
For all $N$ if $f_{N}\rightarrow_{c_N} f_{N+1}$ is a type 1 update, then $\bin(f_N)+1=\bin(f_{N+1})$
\end{lemma}
\begin{proof}
Let $\ell$ be the entry of the insertion in update $f_{N}\rightarrow_{c_N} f_{N+1}$.
Then, $f_{N}(\ell)$ cannot be defined and even, because otherwise $\ell+1$ could also be chosen. Also, $f_{N+1}(\ell)$ is $c_N$ and thus even. On the other hand, since $f_{N}\rightarrow_{c_N} f_{N+1}$ is a type 1 update, we have for $j<\ell$ that $f_{N}(j)$ is defined and even  and $f_{N+1}(j)$ is not defined. No index $>\ell$ changes on an insertion on index $\ell$.
Hence, $\bin(f_{N+1})-\bin(f_{N+1})=2^{\ell}-(1+2+4+\dots+2^{\ell-1})=1$.
\end{proof}

Next, the lemma about type 2 updates.
\begin{lemma}\label{lem:rule2}
For all $N$ if there is an update of type 2 on $(\ell,c)$ at date $N$, then there is a $t<N$ such that 
\begin{enumerate}
\item the update on date $t$ is of type 1
\item the statistics $f_{t+1}$ is equal to $f_{N+1}$ except that $f_{t+1}(\ell)\neq f_{N+1}(\ell)$ and $f_{t+1}(\ell)$ is even
(implying that $\bin(f_{t+1})= \bin(f_{N+1})$ if $c$ is even and $\bin(f_{t+1})> \bin(f_{N+1})$ if $c$ is odd)
\item there are no insertions at index $\ell'$ for any $\ell'>\ell$ between date $t$ and $N$.
\end{enumerate}
Also, $t<N$ is the last date such that there is a type 1 update at that date on index $\ell$.
\end{lemma}
\begin{proof}
Let $t<N$ be the largest date such that there is a type 1 update at that date on index $\ell$.
This is well-defined, since initially, $f_0=\emptyset$ and for the smallest $t'$, such that $\ell\in f_{t'+1}$, we must have that the update at date $t'$ is a type 1 update (by definition).
We see that both $f_{N+1}(i)$ and $f_{t+1}(i)$ are undefined for $i<\ell$ because of the updates on date $N$ and $t$ respectively.
 Also, for all $i>\ell$ such that $f_{t+1}(i)$ or $f_{N+1}(i)$ is defined, we have that both are defined and $f_{t+1}(i)=f_{N+1}(i)$. This is because, if an insertion $(\ell',d)$ is performed for $\ell'>\ell$ on a date $t'$ between $t+1$ and $N$, we have that $f_{t'+1}(\ell)$ becomes undefined and hence, there must be a date $>t'$ such that $\ell$ is inserted again, which would use rule 1 and thus contradict the choice of $t$.
Thus, the statistics match except that $f_{t+1}(\ell)<f_{N+1}(\ell)$ (because each later time we change index $\ell$ we use rule 2 and the entry thus increases). Also, $f_{t+1}(\ell)$ is even because the update on date $t$ is of type 1.
\end{proof}
We also give a corollary.
\begin{corollary}\label{cor:firstrule1}
Fix a number $i>0$.
Consider the smallest date $T$ such that $\bin(f_{T+1})\geq i$. Then the update on date $T$ is of type 1 and $\bin(f_{T+1})= i$
\end{corollary}
\begin{proof}
By minimality of $T$ we get that $\bin(f_{T})< i$ (because $\bin(f_{0})=0$). 
By Lemma~\ref{lem:rule2}, the update on date $T$ has type 1. By Lemma~\ref{lem:rule1}, we thus get that $\bin(f_{T+1})=i$.
\end{proof}

Next, we define even factorization and then show that a long even factorization implies that Anke wins the parity game.

\begin{definition}[Even factorizations]
An \emph{ even factorization} of length $j$ is a sequence $0 \leq t_0<  \ldots < t_j$ such that for every $i\in 0\ldots j-1$,
the maximum of 
$c_{t_i},c_{t_i+1}, \ldots ,c_{t_{i+1}-1}$
is even.
\end{definition}

We next show that long even sequences exists.
\begin{lemma}\label{lem:factor}
For all $N$, there is an even factorization of length at least $\bin(f_{N})$.
\end{lemma}
\begin{corollary}\label{cor:direct}
If Anke wins the statistics game then she wins the parity game.
\end{corollary}
\begin{proof}
By definition of the statistics game,
Anke can enforce the play to reach a statistic $f_{N+1}$ such that $k\in\dom(f_{N+1})$.

If $N$ is chosen minimal then $f_N\to_{c_N} f_{N+1}$ is an update of type 1 by Corollary~\ref{cor:firstrule1} on entry $k$. Hence, $f_{N+1}$ is defined on $k$ and $f_{N+1}(k)$ is even. This implies that $\bin(f_{N+1})\geq 2^k$.
According to Lemma~\ref{lem:factor},
such a play has an even  factorization $t_0 < t_1 <\ldots < t_j$
of length $\bin(f_{N+1})\geq 2^k$. 
Since $2^k$ is $>$ than twice the number of vertices, 
the play loops on the same vertex at some dates $t_i$ and $t_{i'}$, while having the same current player, with $0\leq i< i' \leq j$.
By definition of even factorizations, the maximal priority on this loop is even. Thus Boris has no positional winning strategy in the parity game (because if he had followed it, no loop can have even maximal priority),
and since parity games are positional~\cite{zielonka:1998},
Boris has no winning strategy at all in the parity game.
\end{proof}

\begin{proof}[Proof of Lemma~\ref{lem:factor}]
Consider a fixed $N$. Let $x=\bin(f_{N})$. 
We will show that the following sequence $t_1,\dots,t_x$ is an even factorization.

For ease of notation, let $t_{x+1}=N+1$ (note that $t_{x+1}$ is not part of the even factorization).
For all $j\leq x$, let $t_j<t_{j+1}$ be the last date $T$ using rule 1 such that $\bin(f_{T+1})=j$. 

\smallskip\noindent{\bf Sequence is well-defined.}
This sequence is well-defined because (1)~on the first date $T$ where $\bin(f_{T+1})\geq j$ we use rule 1 and $\bin(f_{T+1})= j$, by Corollary~\ref{cor:firstrule1}; and (2)~$\bin(f_{t_{j+1}})=j$ (and hence a date $T<t_{j+1}$ exists where $\bin(f_{T+1})\geq j$), which is true for $j=x$ by definition of $t_{x+1}$ and otherwise follows from Lemma~\ref{lem:rule1} because we use rule 1 on date $t_{j+1}$ for $j<x$.

\smallskip\noindent{\bf Sequence is an even factorization.}
Consider some fixed $i<x$. We will argue that the maximum priority $c$ in $c_{t_i},c_{t_i+1}, \ldots ,c_{t_{i+1}-1}$ is even.
We will do so using contradiction, so assume that $c$ is odd.
Let $T\in \{t_i,t_i+1,\dots,t_{i+1}-1\}$ be the smallest date such that $c$ is seen on that date and let $\ell$ be the index changed on that date. Note that $T> t_i$, since we use rule 1 on date $t_i$ which requires an even number.

\begin{claim}\label{cla:ell}
The number $\ell$ is well-defined and $\bin(f_{T+1})<i=\bin(f_{t_i+1})$.
\end{claim}
\begin{proof}
Let $\ell'$ be the index inserted at date $t_i$. Let $\ell''$ be the largest index inserted at a date in $t_i,\dots,T-1$. 
By definition of $\ell''$, we have that $f_{T}(\ell'')$ is defined and by definition of $T$ and $c$ we have that $f_{T}(\ell'')<c$.
Thus, we can perform a type 2 insertion of $c$ at $\ell''$ and hence $\ell$ is well-defined.
Thus, $f_{T+1}(i)$ is odd or undefined for $i\leq \ell$, and $\dome(f_{T+1})\cap \ell+1\dots k=\dome(f_{t_i+1})\cap \ell+1\dots k$ because no such entry has changed between those two dates.
On the other hand $f_{t_i+1}(\ell')$ is even since a rule 1 update was used on that index on that date.
\end{proof}

Let $T'\in \{T,\dots, t_{i+1}-1\}$ be the first date such that $\bin(f_{T'+1})\geq i$. 
This is well-defined because we have that $\bin(f_{t_{i+1}})=i$ by Lemma~\ref{lem:rule1} (since we use rule 1 on date $t_{i+1}$). Clearly $T'>T$ since $\bin(f_{T+1})<i$ by Claim~\ref{cla:ell}. This also implies that $\bin(f_{T'})<i$.  We must thus make an update on date $T'$.
We cannot make an update of type 1 on date $T'$, because $\bin(f_{T'})<i\leq \bin(f_{T'+1})$ would then imply that $\bin(f_{T'+1})=i$ by Lemma~\ref{lem:rule1}, which contradicts the choice of $t_i$ (since $t_i<T<T'<t_{i+1}$ as noted).
We next argue that the update on date $T'$ cannot be of type 2 either which contradicts that an update have either type 1 or 2, shows that $c$ must be even and thus completes the proof of the lemma.

\begin{claim}\label{cla:rule1T'}
The update on date $T'$ is not of type 2
\end{claim}

\begin{proof}
Assuming towards contradiction that rule 2 is used on date $T'$.
Let $(\ell',c')$ be the update performed on date $T'$. Since
$\bin(f_{T'})<i\leq \bin(f_{T'+1})$, we have that $c'$ is even. 
We will argue that there can be no such $\ell'$. 
Let $\ell''$ be the largest index changed between date $T$ and date $T'$, both included. We thus have that $\ell''\geq \ell,\ell'$.
 We can apply Lemma~\ref{lem:rule2} and see that there is $t$ such that $t<T'$ and such that $\bin(f_{j+1})=\bin(f_{T'+1})$, because $c'$ is even.
 We thus just need to argue that $t\geq T$ to contradict that $T'$ is the first date in $\{T,\dots, t_{i+1}-1\}$ where $\bin(f_{T'+1})\geq i$.

\smallskip\noindent{\bf If $\bm{\ell''>\ell'}$.} Note that this is especially the case if $\ell>\ell'$. We see that $t\geq T$ because there is no insertion between date $t$ and $T'$ at a higher index than $\ell'$ by Lemma~\ref{lem:rule2}. This contradicts the choice of $T'$.

\smallskip\noindent{\bf Otherwise, if $\bm{\ell''=\ell'\geq \ell}$.} In this case $f_{T+1}(\ell')$ is either not defined or at least $c$. This is because if $f_{T}(\ell')$ was defined and smaller than $c$, then it would be changed on date $T$. 
We have that $f_{T'}(\ell')$ is defined and $<c'\leq c$ because otherwise we could not use rule 2 on date $T'$ and insert into $\ell'$.
Consider the first date $t'\geq T$  such that $f_{t+1}(\ell')$ is defined and $<c$.  Hence $t'\leq T'$. 
To lower an entry or make it defined we must use rule 1 on that entry and thus, we use rule 1 on date $t'$ on entry $\ell'$. 
Hence $t'\neq T'$ (because we use rule 2 on date $T'$) and thus $t'<T'$.
But then $t\geq t'\geq T$ because $t<T'$ is the last date on which rule 1 was used on index $\ell'$ by Lemma~\ref{lem:rule2}. This contradicts the choice of $T'$.
\end{proof}
\end{proof}

\section{Time complexity of solving statistics games}
\paragraph*{Reachability games}
A reachability game $G$ is a tuple $(V,E,\top)$, where $V$ is a set of $n$ vertices and $E\subseteq V\times V$ is a set of $m$ edges. The vertex $\top\in V$ is a the target vertex.
The play starts in some initial vertex $s$, player 1 and 2 alternatively select a vertex $u\in \{u\mid (v,u)\in E\}$. The play then continues to $u$. If the play is ever in $w$, the game ends and player 1 wins, otherwise player 2 wins.

If player $1$ has a strategy to ensure a win from some vertex $s$,
then $s$ is called a winning vertex.
The classical algorithm for reachability games $G$ is called backward induction and computes in time $O(m)$ the set of winning vertices.

\paragraph*{Statistics game as a reachability game}

Given a parity game $G=(V,E)$, with $M$ priorities, $n$ vertices and $m$ edges, let $k=\lceil\log (n+1)\rceil$ be the maximum index in the corresponding statistics game.
Denote $S_{i,M}$ the set of statistics with $M$ priorities and $i$ being the highest possible index. 

The corresponding statistics game is the
reachability game with vertices $ V\times S_{k-1,M}\cup \{\top\}$.
For every edge $(v,u)\in E$ and statistic update $f\rightarrow_{\val(u)}f'$ with $k\not\in \dom(f)$,
 there is an edge from $(v,f)$ to $(u,f')$ if $k\not\in \dom(f')$
or to $w$ if $k\in \dom(f')$.

\paragraph*{A naïve upper complexity bound}

According to Theorem~\ref{theo:main},
a vertex $s$ is winning in the parity game
if and only if the vertex $(s,\emptyset)$ is winning in the statistics game.

The statistics game has $\leq n|S_{k-1,M}|+1$ vertices and $\leq m|S_{k-1,M}|$ edges
and there is a naïve $(M+1)^{\lceil\log (n+1)\rceil}$ upper bound on $|S_{k-1,M}|$.
This gives a first straightforward upper bound on the complexity of solving parity games:
\[
\mathcal{O}\left(m (M+1)^{\lceil\log (n+1)\rceil}\right)
\leq
\mathcal{O}\left(m M^{1 + \log (n)}\right)
 \enspace.
\]

\paragraph*{Tighter upper complexity bounds}

We give tighter upper complexity bounds,
starting with some bounds on  $|S_{i,M}|$ for all $i,M$.

\begin{lemma}\label{lem:inc_func}
Let $x,y\in \ZZ$.
The number of increasing functions $f:\{1,\dots,x\}\rightarrow \{1,\dots,y\}$ is ${{x+y-1}\choose{x}}$
\end{lemma}
\begin{proof}
Each increasing function $f:\{1,\dots,x\}\rightarrow \{1,\dots,y\}$ has a 1-to-1 correspondence with subsets of size $x$ of $\{1,\dots,x+y-1\}$ as follows: Let $S_f$ be the set $\{f(1),f(2)+1,\dots,f(x)+x-1\}$. Observe that since $f$ is increasing, $f(i)+i<f(i+1)+i+1$ for all $i$. Thus $S_f$ has exactly $x$ elements. On the other hand, every set $S=\{1\leq j_0<\dots< j_x\}$ corresponds to the function $f_S(z)=j_z-z+1$. The function $f_S$ is increasing because $j_i>j_{i-1}$ for all $i$. There are ${{x+y-1}\choose{x}}$ subsets of size $x$ of $\{1,\dots,x+y-1\}$.
\end{proof}

\begin{lemma}\label{lem:partial_inc_func}
Let $x,y\in \ZZ$.
The number of increasing, partial functions $f:\{0,\dots,x\}\rightarrow \{1,\dots,y\}$ is $\sum_{i=0}^{x+1}{{x+1}\choose{i}}\cdot {{i+y-1}\choose{i}}$
\end{lemma}
\begin{proof}
A partial increasing function is a increasing function in its domain. 
For a fixed $i$, there are ${{x+1}\choose{i}}$ domains of size $i$. 
Since each domain of size $i$ corresponds to the domain $\{1,\dots,i\}$ we can apply Lemma~\ref{lem:inc_func} and see that 
there are ${{i+y-1}\choose{i}}$ increasing functions for a fixed domain of size $i$.
Thus, there are  $
\sum_{i=0}^{x+1}{{x+1}\choose{i}}\cdot {{i+y-1}\choose{i}}$
increasing partial functions $f:\{0,\dots,x\}\rightarrow \{1,\dots,y\}$ in total.
\end{proof}

Hence, the time complexity of backwards induction on the statistics game is \[
O(m|S_{k-1,M}|)=O\left(m\sum_{i=0}^{k}{{k}\choose{i}}\cdot {{i+M-1}\choose{i}}\right) \enspace .
\]

\begin{theorem}
Given a parity game with $n$ vertices, $m$ actions and max priority $M$, the winner of each initial vertex can be found in time \[O\left(\min\left(mn^22^M/\sqrt{M\log n},mn^{2.4427...}n^{\log (1+\frac{M}{\log n})}\cdot \left(1+\frac{M}{\log n}\right)\right)\right)\enspace .\] 

Especially, for $M\geq \epsilon\log^2 n$, for some constant $\epsilon>0$, the winner can be found in $O(m\cdot n^{1.4427...}n^{\log (1+\frac{M}{\log n})}\cdot (1+\frac{M}{\log n}))$ time.

For $M=\log n$, the winner can be found in time $O\left(mn^{\log\frac{\sqrt 2+1}{\sqrt 2-1}}\right)= O(mn^{2.5431...})$.
\end{theorem}
\begin{proof}
We will give an upper bound on $O\left(m\sum_{i=0}^{k}{{k}\choose{i}}\cdot {{i+M-1}\choose{i}}\right)$.

Let $g(i)={{k}\choose{i}}\cdot {{i+M-1}\choose{i}}$.
For $i=k$ we have that 
\begin{align*}
g(k)&=\frac{(k-1+M)!}{k!(M-1)!}=\frac{k-1+M}{k}\cdot {{k-2+M}\choose{k-1}}=\frac{k-1+M}{k^2}\cdot g(k-1)
\end{align*}

For $0<i<k$ we have that
\begin{align*}
g(i)=
\frac{k!}{i!(k-i)!}\cdot 
\frac{(i-1+M)!}{i!(M-1)!}=
\frac{(k-i)(i-1+M)}{i^2}\cdot g(i-1) \enspace .
\end{align*}

Observe that $2^k=2^{\lceil{\log (n+1)}\rceil}<2^{\log (n+1)+1}=2(n+1)$.

A trivial bound on ${{y}\choose{x}}$ for all $x,y$ is $y^x/x!$. We thus get  using Stirling's approximation  that
\begin{align*}
g(k)&\leq (k+M-1)^{k}/k!
< 1/2 e^{k \ln (k+M-1)}/e^{(k+1/2)\ln k-k}\\
&= 1/2 e^{k \ln (k+M-1)-(k+1/2)\ln k+k}
= 1/2 e^{k (1-\ln k+\ln (k+M-1))-(\ln k)/2}\\
&=1/2 e^{k \ln ((e(k+M-1))/k)-(\ln k)/2}=k^{-1/2}\cdot (e(k+M-1)/k)^k\\
&=k^{-1/2}\cdot (2(n+1))^{\log e+\log (1+(M-1)/k)}=O(k^{-1/2}\cdot n^{1.4427...}n^{\log (1+\frac{M-1}{\log n})}\cdot (1+\frac{M-1}{\log n}))
\end{align*}

We first consider the case where $M\geq \epsilon k^2$ for some constant $\epsilon>0$.
Observe that $g(k)$ is a factor $\epsilon$ of $g(k-1)$ for this choice of $M$.
Also, for $0<i<k$ we have that $g(i)/g(i-1)>(k-i)\epsilon$.
Thus, $g(i)$ is decreasing geometrically (with a constant factor of at most $1/\epsilon$) for $k-1/\epsilon>i$ and increasing below that. But, $1/\epsilon$ is a constant and thus, $\sum_{i=0}^k g(i)$ is $O(g(k))=O({{k-1+M}\choose{k}})=O(k^{-1/2}\cdot n^{\log e+\log (1+(M-1)/k)}\cdot (1+\frac{M-1}{\log n}))$. Hence, the time complexity is $O(m\log^{-1/2} n\cdot n^{1.4427...}n^{\log (1+\frac{M-1}{\log n})}\cdot (1+\frac{M-1}{\log n}))$ in this case.

Next we consider smaller values of $M\geq k+1$.
Note that ${{y}\choose{x}}$ is geometrically increasing for a fixed $y$ for $x<y/2$ and geometrically decreasing for $x>y/2$. Also, ${{y}\choose{y/2}}\approx 2^y/\sqrt{y}$.

Thus, the time complexity is \begin{align*}
O(m|S_{k-1,M}|)&=O\left(m\sum_{i=0}^{k}{{k}\choose{i}}\cdot {{i-1+M}\choose{i}}\right)=O\left(m\cdot 2^k\sum_{i=0}^{k}{{i-1+M}\choose{i}}\right)\\
&=O\left(m\cdot 2^k/\sqrt{k}{{k-1+M}\choose{\min(\frac{k-1+M}{2},k)}}\right)\\
&=O\left(\min\left(mn^22^M/\sqrt{M\log n},mn^{2.4427...}n^{\log (1+\frac{M-1}{\log n})}\cdot (1+\frac{M-1}{\log n})\right)\right)
 \enspace .
\end{align*}

Note that the above argument basically finds the maximum of ${{k}\choose{i}}$ and ${{i-1+M}\choose{i}}$ independently, i.e. without using that it is the same $i$. Thus, one can give better bounds for especially specific values of $M$ as a function of $k$.
We see that $g(i)$ keeps increasing until $g(i)/g(i-1)\leq 1$. Let $i_*$ be the smallest such $i$.
\begin{align*}
&1\geq g(i_*)/g(i_*-1)\Leftrightarrow 1\geq \frac{(k-i_*)(i_*+M-1)}{i_*^2}\\
\Leftrightarrow &i_*^2 \geq (k-i_*)(i_*+M-1)\Leftrightarrow i_*^2 \geq k(M-1) +ki_*-(M-1)i_* -i_*^2\\
\Leftrightarrow &0 \geq k(M-1) +ki_*-(M-1)i_* - 2i_*^2\Leftrightarrow 0 \leq -k(M-1) -ki_*+(M-1)i_* + 2i_*^2\\
\Leftrightarrow & i_* \geq \frac{k-(M-1)+ \sqrt{(M-k-1)^2+8k(M-1)}}{4}\enspace .
\end{align*}

We then get an upper bound on $O(m|S_{k-1,M}|)$ of \[
O(m\cdot \sum_{i=0}^k g(i))=O(m\cdot k g(i_*)) \enspace .
\]
This bound is accurate upto a factor of $k=O(\log n)$.

Thus, for instance, for $M=k+1$, we have that $i_*= \frac{k}{\sqrt{2}}$. Inserting this into $O(m\cdot k g(i_*))$ we get that \begin{align*}
m\cdot k g(i_*)&=m\cdot k g(i_*)=m\cdot k \cdot {{k}\choose{i_*}}\cdot {{k+i_*}\choose{i_*}}\\
&=m\cdot k \cdot \frac{k!}{i_*!(k-i_*)!}\frac{(k+i_*)!}{i_*!k!}=
m\cdot k \cdot \frac{(k+i_*)!}{(i_*!)^2(k-i_*)!}\\
&< m\cdot k \cdot e^{(k+i_*+1/2)\ln(k+i_*)-(k+i_*)-2((i_*+1/2)\ln(i_*)-i_*)-((k-i_*+1/2)\ln(k-i_*)-(k-i_*))}\enspace ,
\end{align*}
where we used Stirling's approximation for the inequality.
We next consider the exponent of $e$. \begin{align*}
&(k+i_*+1/2)\ln(k+i_*)-(k+i_*)-2((i_*+1/2)\ln(i_*)-i_*)\\&-((k-i_*+1/2)\ln(k-i_*)-(k-i_*))\\
=&i_*(\ln(k+i_*)+\ln(k-i_*)-2\ln(i_*))
+k(\ln(k+i_*)-\ln(k-i_*))\\
&+1/2(\ln(k+i_*)-2\ln(i_*)-\ln(k-i_*))\\
=&i_*\left(\ln\left(\frac{(k+i_*)(k-i_*)}{i_*^2}\right)\right)
+k\left(\ln\left(\frac{k+i_*}{k-i_*}\right)\right)
+1/2\left(\ln\left(\frac{k+i_*}{(i_*^2)(k-i_*)}\right)\right)\\
=&i_*\left(\ln\left(\frac{k^2-k^2/2}{k^2/2}\right)\right)
+k\left(\ln\left(\frac{k+k/\sqrt 2}{k-k/\sqrt 2}\right)\right)
+1/2\left(\ln\left(\frac{k+k/\sqrt 2}{k^2/2(k-k/\sqrt 2)}\right)\right)\\
=&k\left(\ln\left(\frac{1+1/\sqrt 2}{1-1/\sqrt 2}\right)\right)
+1/2\left(\ln\left(\frac{1+1/\sqrt{2}}{k^2(1-1/\sqrt{2}}\right)\right)\\
\end{align*}
Inserting it back into the earlier expression, we get that the time complexity is
\begin{align*}
O\left(m\cdot k \cdot e^{k(\ln(\frac{1+1/\sqrt 2}{1-1/\sqrt 2}))
+1/2(\ln(\frac{1+1/\sqrt{2}}{k^2(1-1/\sqrt{2})}))}\right)
&=O\left(m\left(\frac{1+1/\sqrt 2}{1-1/\sqrt 2}\right)^k\right)\\
&=O\left(mn^{\log\frac{\sqrt 2+1}{\sqrt 2-1}}\right)= O(mn^{2.5431...})
\end{align*}

\end{proof}


\begin{thebibliography}{1}

\bibitem{calude}
C.~S. Calude, S.~Jain, B.~Khoussainov, W.~Li, and F.~Stephan.
\newblock Deciding parity games in quasipolynomial time.
\newblock Technical report, CDMTCS, October 2016.
\newblock URL:
  \url{https://www.cs.auckland.ac.nz/research/groups/CDMTCS/researchreports/index.php?download&paper_file=631}.

\bibitem{zermelo}
E.~Zermelo.
\newblock Uber eine anwendung der mengenlehre auf die theorie des schachspiels.
\newblock In {\em Proc. of the Fifth International Congress of Mathematicians},
  volume~II, pages 501--504. Cambridge University Press, 1913.

\bibitem{zielonka:1998}
Wies{\l}aw Zielonka.
\newblock Infinite games on finitely coloured graphs with applications to
  automata on infinite trees.
\newblock {\em TCS}, 200:135--183, 1998.

\end{thebibliography}
\end{document}